\documentclass[a4paper,12pt]{article}

\usepackage{amscd}
\usepackage{color}

\usepackage{graphicx} 
\advance \textwidth -60truemm\relax

\advance\oddsidemargin -1truein\relax

\evensidemargin=\oddsidemargin

\advance\textheight -60truemm\relax
\advance\topmargin -1truein\relax
\unitlength\textwidth
\divide\unitlength by 150\relax

\usepackage{amsmath,amssymb}
\usepackage{bm}
\bmdefine{\NNN}{N}
\bmdefine{\ZZZ}{Z}
\bmdefine{\RRR}{R}
\bmdefine{\CCC}{C}
\bmdefine{\bbb}{b}
\bmdefine{\uuu}{u}
\bmdefine{\vvv}{v}
\bmdefine{\www}{w}
\bmdefine{\eee}{e}
\bmdefine{\xxx}{x}
\bmdefine{\XXX}{X}
\bmdefine{\zerovec}{0}

\numberwithin{equation}{section}
\newtheorem{thm}[equation]{Theorem}

\newtheorem{example}[equation]{Example}
\newtheorem{lemma}[equation]{Lemma}

\newtheorem{definition}[equation]{Definition}
\newtheorem{prop}[equation]{Proposition}
\newtheorem{remark}[equation]{Remark}

\newcommand{\bigzerou}{\smash{\lower1.7ex\hbox{\bg 0}}}

\newcommand{\bigastu}{\smash{\lower1.7ex\hbox{\bg *}}}

\begin{document}
\begin{center}
{\bf \large  Tests of Non-Equivalence among Absolutely Nonsingular Tensors through Geometric Invariants\\}
\end{center}
\begin{center}
Sakata, T.$^1$, Maehra, K.$^2$, Sasaki,T.$^3$, Sumi, T. $^{1}$, Miyazaki, M.$^{4}$ and Watanabe Y.$^5$  \\
\vspace*{2mm}
\end{center}
Department of Design Human Science, Kyushu University$^{1}$\\ 
School of Design, Kyushu University$^{2}$.\\
Department of Mathematics, Kobe University$^{3}$.\\
Department of Mathematics, Kyoto University of Education$^{4}$\\
Research Institute of Information Technology, Kyushu University$^{5}$
\section{Introduction} Tensor data analysis has successfully developed in various application fields, which is useful to  seize multi-factor dependence.  
An $n \times n \times p$ tensor is a multi-array datum $T=(T_{ijk}),$ where $1 \leq i,j \leq n$ and $1 \leq k \leq p.$ A type $n \times n \times p$ tensor $T$ is denoted by $T=(A_{1};A_{2};\cdots;A_{p})$ where $A_{i}$ denote $n \times n$ matrices. An $n \times n \times p$ tensor $T$ is said to be of rank $1$ if there is a vectors $\bm{a}=(a_{1},a_{2},...,a_{n}),$ $\bm{b}=(b_{1},b_{2},...,b_{n})$, and  $\bm{c}=(c_{1},c_{2},...,c_{p})$ such that $T_{ijk}=a_{i}b_{j}c_{k}$ for all $i,j,k,$
and the rank of a tensor $T$ is defined as the minimum of  the integer $r$ such that $T$ can be expressed as the sum of $r$ rank-one tensors. The maximal rank of  all tensors of type $n \times n \times p$ are also defined in obvious fashion and denoted by ${\rm maxrank(n,n,p)}$.The rank of a tensor describes the complexity of a tensorial datum and the maximal rank describes the model complexity of a class of tensors of a given type, and so they are very important concepts in both applied and theoretical fields. Therefore, the rank and maximal rank determination problems have attracted the interest of many researchers, for example, Kruskal \cite{Kruskal}, ten-Berge \cite{tenBerge}, and Common et al, \cite{Common}, etc. and now have being investigated intensively( for a comprehensive survey, see Kolda et al.\cite{Kolda}). Atkinson et al.(\cite{Atkinson1} and \cite{Atkinson2} ) claimed that ${\rm maxrank}(n,n,3) \leq 2n-1$. Here we introduce an important class of tensors.
\begin{definition}
A real tensor $T=(A_{1};A_{2};,\dots;A_{p})$ is said to be absolutely nonsingular if $x_{1}A_{1}+x_{2}A_{2}+\cdots +x_{p}A_{p}$ is nonsingular 
for all $(x_{1},x_{2},...,x_{p}) \neq (0,0,...,0).$  
\end{definition}
\begin{remark}
We called absolutely nonsingular tensors as exceptional tensors
in Sakata et al. {\rm \cite{Sakata1},\cite{Sakata2}}.
\end{remark}
 Sumi et al.\cite{Sumi1}  proved the claim of Atkinson et al. over the complex number filed $\mathbb{C}$ without any assumption and proved it over the real number filed $\mathbb{R}$ except the class of absolutely nonsingular tensors. Thus, for the proof of the claim of Atkinson et al. over the real number filed $\mathbb{R},$ it is the first thing to determine all absolutely nonsingular tensors. Absolutely nonsingular tensors are characterized by the determinant polynomial defined below. Searching of absolutely nonsingular tensors was pursued in Sakata et al.  {\rm \cite{Sakata1},\cite{Sakata2}} in this direction.  As well as searching absolutely nonsingular tensors, the equivalence among them under the rank-preserving transformation which is defined below is also important. 
Note that such equivalence relation has also a relation to the SLOCC equivalence of entangled states in the quantum communication (for example, see Chen et al. \cite{Chen}).          
\begin{definition}
For a $n\times n \times p$ tensor $T=(A_{1};A_{2};,\dots;A_{p}),$ the homogeneous polynomial in $x_{1},...,x_{p}$ of degree $n$ 
\begin{equation}
f_{T}(x_{1},x_{2},...,x_{p})=\det(x_{1}A_{1}+x_{2}A_{2}+\cdots +x_{p}A_{p})
\end{equation}
is called  the determinant polynomial of a tensor $T$.
\end{definition}
Then we have the following important characterization.
\begin{thm}\label{positivity}
If $T=(A_{1};A_{2};,\dots;A_{p})$ is absolutely nonsingular, its determinant polynomial $f_{T}(x_{1},x_{2},...,x_{p})$ is a positive definite homogeneous polynomial or negative definite homogeneous polynomial.
\end{thm}
\begin{proof}
 Let $T=(A_{1};A_{2};,\dots;A_{p})$ be absolutely nonsingular, and assume that there are two points $\bm{x}_{0}$ and $\bm{x}_{1}$ such that $f_{T}(\bm{x}_{0})>0$ and   $f_{T}(\bm{x}_{1})<0.$  The line $\ell$ combining the two points $\bm{x}_{0}$ and $\bm{x}_{1}$ must pass through the origin $\bm{0}$, since in the segment $[x_{0},x_{1}]$ there must be $\bm{x}'$ such that $f_{T}(\bm{x}')=0$ and it must be $\bm{0}$ because $T$ is absolutely nonsingular. Let take another point $\bm{x}_{2}$ which is not on the line $\ell$. 
Then, the line passing $\bm{x}_{0}$ and  $\bm{x}_{2}$ does not pass the origin and so $f(\bm{x}_{0}) f(\bm{x}_{2})<0$ is impossible just by the same reason given in the previous sentence.  So, $f(\bm{x}_{0}) f(\bm{x}_{2})>0.$ Next, consider the line passing $\bm{x}_{1}$ and $\bm{x}_{2}$, which also does not pass the origin and   $f(\bm{x}_{1}) f(\bm{x}_{2})<0$. This is also a contradiction. After all, there don't exist points $\bm{x}_{0}$ and $\bm{x}_{1}$ such that $f(\bm{x}_{0})>0$ and $f(\bm{x}_{1})<0.$  This proves Theorem \ref{positivity}.
\end{proof}
It is well known that tensor rank is invariant by typical matrix transformations, say, $p-$, $q-$, and $r-$transformations defined below. So, equivalence relation of two tensors means that they have a same rank. Thus, to study equivalence among tensors is of some importance for rank determination.    
\begin{definition}
For a $n \times n \times p$ tensor $T=(A_{1};A_{2};\cdots;A_{p}),$ the following transformations 
\begin{itemize}
\item[(1)] $T=(A_{1};A_{2};\cdots;A_{p}) \rightarrow T'=(PA_{1};PA_{2};\cdots;PA_{p}) $  by  an $n \times n$ matrix $P \in GL(n),$ 
\item[(2)] $T=(A_{1};A_{2};\cdots;A_{p}) \rightarrow T'=(A_{1}Q;PA_{2}Q;\cdots;A_{p}Q)$ by  an $n \times n$ matrix $P \in GL(n),$ 
\item[(3)]  $T=(A_{1};A_{2};\cdots;A_{p}) \rightarrow T'=(R_{11}A_{1}+R_{12}A_{2}+R_{13}A_{3}; R_{21}A_{1}+R_{22}A_{2}+R_{23}A_{3};R_{31}A_{1}+R_{32}A_{2}+R_{33}A_{3})$  by  a $p \times p$ matrix $P \in GL(n)$ 
\end{itemize}
are called as  $p-$, $q-$, and $r-$transformations and denoted by $T \rightarrow_{p} T',  T \rightarrow_{q} T' and T \rightarrow_{r} T'$ respectively. Further, if $T_{1} \rightarrow_{p} T_{2}$, the $T_{1}$ and $T_{2}$ are said to be in the $p-$equivalence.  $q-$ and $r-$equivalence are defined analogously.
\end{definition}
\begin{definition}
 Let $T_{1}=(A_{1};A_{2};\cdots;A_{p})$ and  $T_{2}=(B_{1};B_{2};\cdots;B_{p})$ be  two $n \times n \times p$ tensors. If there is a sequence of $\{T_{i} \}$ starting from $T_{1}$ and ending at $T_{2},$  in which $T_{i}$ and $T_{i+1}$ are in the relation of $p-$, or $q-$, or $r-$equivalence, then $T_{0}$ and $T_{1}$ are said to be equivalent. 
\end{definition}
Now we can reduce the equivalence relation into a more simple one by the following lemma.
\begin{lemma}\label{comutative}
$p-$, $q-$ and $r-$transformations are mutually commutative. 
\end{lemma}
\begin{proof}
For simplicity, we prove for $p=3$, however, the proof  is  similar for a general $p$. First we prove the commutativity of $p-$transformation and  $r-$transformation. 
Let 
$$
T_{1} \rightarrow_{p} T_{2} \rightarrow_{r} T_{3} \ \ and \ \ T_{1} \rightarrow_{r}T^{'}_{2} \rightarrow_{p}T^{'}_{3}
$$ 
We will show that $T_{3}=T^{'}_{3}.$ Let $T_{1}=(A_{1};A_{2};A_{3})$ and $P=(p_{ij})$ and $R=(r_{ij}).$  Then,
$$
T_{2}=(PA_{1};PA_{2};PA_{3}) 
$$
and 
$$
T_{3}=(r_{11}PA_{1}+r_{12}PA_{2}+r_{13}PA_{3}; r_{21}PA_{1}+r_{22}PA_{2}+r_{23}PA_{3};  r_{31}PA_{1}+r_{32}PA_{2}+r_{33}PA_{3})
$$
On the other hand 
$$
T^{'}_{2}=(r_{11}A_{1}+r_{12}A_{2}+r_{13}A_{3}; r_{21}A_{1}+r_{22}A_{2}+r_{23}A_{3};  r_{31}A_{1}+r_{32}A_{2}+r_{33}A_{3})
$$
and 
$$
T^{'}_{3}=(r_{11}PA_{1}+r_{12}PA_{2}+r_{13}PA_{3}; r_{21}PA_{1}+r_{22}PA_{2}+r_{23}PA_{3};  r_{31}PA_{1}+r_{32}PA_{2}+r_{33}PA_{3})
$$
Thus, $T_{3}=T^{'}_{3},$ and this means the commutativity of $p$- and  $r$-transformations. 
The commutativity of $q$- and $r$-transformations are proved similarly. $p$- and $q$-transformations are obviously commutative. 
This proves Lemma \ref{comutative}.    
\end{proof}
Note that in this paper we consider three cases of  (1) $P, Q \in GL(n)$ and  $R \in GL(p)$ and (2)  $P, Q \in GL(n)$ and $R \in SL(p).$
and (3) $P, Q \in SL(n)$ and $R \in SL(p).$ The first  is called $GL(p)$-equivalence or simply equivalence, and the second  is called $SL(p)$-equivalence in short. The third case is called, in a full term, $SL(n) \times SL(n) \times SL(p)$-equivalence. Lemma \ref{comutative} implies the following theorem.
\begin{thm}
$T_{1}$ and $T_{2}$ are $GL(p)$-equivalent if and only if there is a set of  $p$-transformation, $q$-transformation and $r$-transformation such that 
$$T_{1} \rightarrow_{p} T^{'} \rightarrow _{q} T^{"} \rightarrow_{r} T_{2}$$ 
\end{thm}
Thus, the equivalence problem of tensors  $T_{1}=(A^{(1)}_{1}:,,,:A^{(1)}_{p})$ and $T_{2}=(A^{(2)}_{1}:,,,:A^{(2)}_{p})$  is reduced to the problem whether the following system of algebraic equations for $P$, $Q$ and $R$ can have a solution or not.
\begin{equation}
A^{(2)}_{i}=P( \{ \sum_{j=1}^{p}r_{ij}A^{(1)}_{j}\}Q, i=1,2,...,p
\end{equation}
These algebraic equations have too many variables to solve even when the size of matrices $A_{i}$ is moderate. So, in this paper, we propose to see the problem through the determinant polynomial. Then, though  we necessarily have to discard the sufficiency part of the problem, however, the problem becomes concise and tractable one by the following proposition.    
\begin{prop}\label{propequivalence}
If $T_{1}$ and $T_{2}$ are $GL(p)$-equivalent, it holds that there is a constant $c \in \mathbb{R}$ and a $p \times p$ nonsingular matrix $R \in GL(p)$ such that
\begin{equation}\label{equivalentequation}
f_{T_{2}}(\bm{x})=cf_{T_{1}}(\bm{x}R)
\end{equation}
\end{prop}
So, we can say that
\begin{prop}
For two tensors, if the equation (\ref{equivalentequation}) does not hold for any constant $c \in \mathbb{R}$ and any matrix $R\in GL(p)$, they are not $GL(p)$-equivalent.    
\end{prop}
Though the reduced equation (\ref{equivalentequation}) happens to be solved algebraically in some cases. However, it is still hard to solve, in general, a system of algebraic  equation with too many variables. In fact, we need to decide whether a system of $ \frac{(n+1)(n+2)}{2}$ homogeneous equations with $(p^2+1)$ variables of degree $n$ have a solution or not. So, in this paper, we avoid to solve the problem algebraically and propose to 
attack the problem from a geometric view point, that is, we propose to test non equivalence by checking whether the two surfaces of the determinant polynomials of $T_{1}$ and $T_{2}$  have a same geometric invariants, or not. Here, multi-linear algebra and differential geometry intersect through the widow of determinant polynomials. \\
The first aim of this paper is to show theoretically that differential geometric invariants are useful as testers of non-equivalence among absolutely nonsingular tensors. The second aim is to show that we can calculate the values of the invariants with enough accuracy. Third, we compare the values of invariants calculated by the lattice method and by the t-design method.  And it is shown that the lattice point method gives more stable values than the t-design method.  As $SL(p)$-invariant, we consider first the volume enclosed by the constant surface and then we consider the affine surface area, and thirdly  we consider the $L^{p}$ affine surface area of convex body. Affine surface area was studied by Blaschke \cite{BL} and extended to $L_{p}$ affine surface area by Lutwak \cite{Lut}, (also see  Leichtwess \cite{Leicht}). As for a valuation theory of $L_{p}$ affine surface area,  see the recent papers by Ludwig \cite{Ludwig2} and Ludwig and Reitzer \cite{LudwigReit}. Finally, as a general reference of affine differential geometry, see K. Nomizu and T. Sasaki \cite{NS}.  \\
This paper is organized as follows. In Section 2, we show how to parametrize the constant surface of a determinant polynomial and in Section 3, we review briefly some definitions from  differential geometry. In Section 4, we argue rough $SL(p)$-invariants. In Section 5, we deal with $SL(p)$-invariant.  In the first subsection, we introduce the valuation theory for the set of convex bodies and in the second subsection, we argue a volume of the region enclosed by a constant surface as an $SL(p)$-invariant. In the third subsection, we argue the affine surface 
as a $SL(p)$-invariant. In Section 6, we consider the generalized affine surface, that is, $L_{p}$ affine surface area, especially centro-affine surface as a $GL(p)$-invariant. In Section 7,  we review the theory of spherical t-design briefly and give a theorem important for approximate calculation of our proposed invariants. 
In Section 8, we give numerical values of the invariants calculated by the lattice method and t-design method. It is shown numerically that the proposed invariants is usefull to discriminate non equivalence. In Section 9, the conclusion is given. Finally note that in the following we consider mainly the case of $n=4$ and $p=3,$ though some statements are given for general $n$ and $p.$ One reason is that absolutely nonsingular tensors are not so easy to obtain for general cases and the second reason is because it is easy to see that our method is also available for general cases. The study of much  higher values of $n$ and $p$ will be given in the future work.   
\section{Parametrization of constant surface}
The determinant polynomial of a $4 \times 4 \times 3$ tensor $T=(A_{1};A_{2};A_{3})$, i$f_{T}(x,y,z)=det(xA_{1}+yA_{2}+zA_{3} ),$ is a homogeneous polynomial of three variables with degree 4. We are concerned with the integral invariants of the constant surface $\partial\Omega_{T}=\{(x,y,z)|f_{T}(x,y,z)=1\}$ for the special linear group $SL(3)$ and the general linear group  $GL(3)$.  To get such invariants, we need to parametrize this surface by the usual spherical coordinate,
\begin{eqnarray}
x&=&r\sin s\cos t=r\Phi_{x}(s,t) \\
y&=&r\sin s\sin t=r\Phi_{y}(s,t) \\ 
z&=&r\cos s=r\Phi_{z}(s,t) ,
\end{eqnarray}
where  $0<s<\pi, 0<t<2\pi.$ Let $\bm{x}$ denote the point $(x,y,z)$ on the surface. Putting these into the equation $f_{T}(x,y,z)=1,$ we have 
\begin{equation} 
r^{4}=\frac{1}{p(s,t)}, 
\end{equation}
where 
\begin{equation} 
p(s,t)=f_{T}(\Phi_{x}(s,t),\Phi_{y}(s,t), \Phi_{z}(s,t)).
\end{equation}
And so, 
\begin{equation} \label{parametric}
\bm{x}=\frac{1}{p(s,t)^{1/4}} \left( \Phi_{x}(s,t),\Phi_{y}(s,t), \Phi_{z}(s,t) \right). 
\end{equation}
This equation (\ref{parametric}) gives a parametric representation of the constant surface $\partial \Omega_{T}$.  Then, the following is a starting point of this research of the constant surface.
\begin{thm}\label{compact}
The constant surface of the determinant polynomial of an absolutely nonsingular tensor is a compact set in $\mathbb{R}^{3}$ without self-intersection.
\end{thm}
\begin{proof}
Without loss of generality, we assume that $f_{T}(\bm{x})$ is positive definite.  If  $\bm{x} \in \partial\Omega_{T}$, for any $0<r<1$ and $r>1,$ $r\bm{x}$ is not in $\partial \Omega_{T}.$ That is, the constant surface is of a star-shaped. This proves that the surface has not any self intersection. Since $p(s,t)$ is continuous on the unit sphere it takes a positive minimum and a positive maximum. So, $\bm{x}$ in the equation \ref{parametric} is bounded, which implies the compactness of the constant surface. This completes the proof of Theorem \ref{compact}.    
\end{proof}
The following 8 figures are examples of the constant surfaces of $ 4 \times4 \times 3$ absolutely nonsingular tensors.


\begin{figure}[htbp]
\begin{center}
\includegraphics[width=3cm, bb= 0 0 150 150 ]{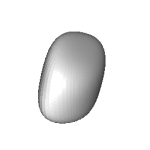}
\includegraphics[width=3cm, bb= 0 0 150 150 ]{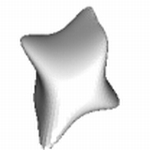}
\includegraphics[width=3cm, bb= 0 0 150 150 ]{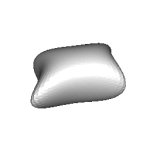}
\includegraphics[width=3cm, bb= 0 0 150 150 ]{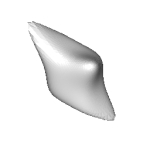}
\includegraphics[width=3cm, bb= 0 0 150 150 ]{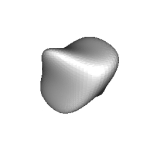}
\includegraphics[width=3cm, bb= 0 0 150 150 ]{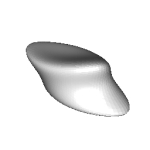}
\includegraphics[width=3cm, bb= 0 0 150 150 ]{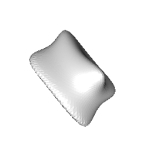}
\includegraphics[width=3cm, bb= 0 0 150 150 ]{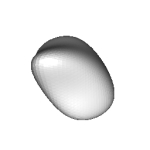}
\caption{Constant surfaces of the determinant polynomials of tensors：$No1, 3,10,20,99,119,207,237.$}
\end{center}
\end{figure}
Note that the numbering of tensors is based on the list of nonsingular tensors with elements consisting only of $-1, 0, 1$ found by us.  
Each figure corresponds to the following tensor and its determinant function respectively.
$$
T_{1}=\left( \begin{array}{cccc}  
 1 &  0 & 0 & 0  \\
  0 & 1 & 0 & 0\\
 0 &  0 & 1 & 0 \\
 0  & 0 &  0 & 1\\
\end{array}\right)
:
\left( \begin{array}{cccc}  
 -1 &   0 &   1 &   1  \\
  1  & -1 &   1 &  -1\\
 -1 &  -1 &  -1 &  -1 \\
 0   & 1 &   1 &   0\\
\end{array}\right)
:
\left( \begin{array}{cccc}  
    0  &  1 &  -1 &   1 \\
   0  & -1 &   1 &   1\\
    0 &  -1 &   0 &  -1\\
   -1 &  -1 &   1 &   1\\
   \end{array}\right),
   $$
\begin{eqnarray*}
f_{T_{1}}(x,y,z)&=&x^4+6y^4+2z^4+-3x^3y-8xy^3-3xz^3+5z^3y+\\
&& 7x^2y^2+3x^2z^2+8z^2y^2-2xy^2z-8xyz^2.
\end{eqnarray*}
$$
T_{3}=\left( \begin{array}{cccc}  
 1 &   0 &   0 &   0  \\
  0 & 1 &  0 & 0\\
 0 &  0 &  1 &  0 \\
 0   & 0 &   0 &   1\\
\end{array}\right)
:
\left( \begin{array}{cccc}  
 0 &   1 &   0 &   0  \\
  0  & -1 &   0 &  1\\
 1 &  0 &  - &  0 \\
 -1   & 1 &   -1 &   1\\
\end{array}\right)
:
\left( \begin{array}{cccc}  
    1  &  1 &  1 &   1 \\
   -1  & 0 &   0 &   -1\\
    0 &  -1 &   0 &  0\\
   -1 &  1 &  -1 &   1\\
   \end{array}\right),
   $$
  \begin{eqnarray*} 
f_{T_{3}}(x,y,z)&=&x^4+3y^4+6z^4-3x^3z+2xy^3+4y^3z\\
&&-7xz^3-6z^3y+x^2y^2+5x^2z^2-5z^2y^2+4xy^2z-x^2yz+2xyz^2.
\end{eqnarray*}
$$
T_{10}=\left( \begin{array}{cccc}  
 1 &   0 &   0 &   0  \\
  0 & 1 &  0 & 0\\
 0 &  0 &  1 &  0 \\
 0   & 0 &   0 &   1\\
\end{array}\right)
:
\left( \begin{array}{cccc}  
   0   & 1 &   0 &   0 \\
    0  & -1 &   0 &   1\\
    1   & 0 &   1 &   0\\
    0   & 0 &  -1 &  -1\\
 \end{array}\right)
 :
\left( \begin{array}{cccc}     
    1  &  1 &   1 &   1\\
  -1  &  0 &   0 &  -1\\
    0  & -1 &   0 &   0\\
   -1  &  1 &  -1 &   1\\
\end{array}\right),
   $$
  \begin{eqnarray*} 
 && f_{T_{10}}(x,y,z)=x^4+y^4+2z^4-x^3y+2x^3z+xy^3+2xz^3\\
&&-z^3y-x^2y^2+4x^2z^2-2xy^2z-2x^2yz-xyz^2.\\
  \end{eqnarray*}
$$
T_{20}=\left( \begin{array}{cccc}  
 1 &   0 &   0 &   0  \\
  0 & 1 &  0 & 0\\
 0 &  0 &  1 &  0 \\
 0   & 0 &   0 &   1\\
\end{array}\right)
:  
\left( \begin{array}{cccc}  
      
 -1 &   0 &   1 &  -1\\
    0  &  0 &   0 &   1\\
    1  & -1 &   0&    0\\
    1  &  0 &   0 &   0\\
    
 \end{array}\right)
 : 
\left( \begin{array}{cccc}     
    1  &  0 &   1 &   1\\
  1  &  1 &   1 &  -1\\
    0  & 0 &   0 &   1\\
   1  &  -1 &  -1 &   -1\\
\end{array}\right),
   $$
 \begin{eqnarray*} 
 && f_{T_{20}}(x,y,z)= x^4+y^4+2z^4-x^3y+x^3z+y^3z-xz^3+2z^3y+\\
 &&-x^2z^2+6z^2y^2+xy^2z+x^2yz+2xyz^2.
 \end{eqnarray*}
$$
T_{99}=\left( \begin{array}{cccc}  
 1 &   0 &   0 &   0  \\
  0 & 1 &  0 & 0\\
 0 &  0 &  1 &  0 \\
 0   & 0 &   0 &   1\\
\end{array}\right)
:   
\left( \begin{array}{cccc}  
      
 0 &   -1 &  - 1 &  -1\\
    1 &  1 &   1 &  - 1\\
    1  & -1 &   1&    0\\
    -1  &  1 &   -1 &   1\\
    
 \end{array}\right)
 :  
\left( \begin{array}{cccc}     
    0  &  1 &   0 &   0\\
  0  &  -1 &   1 &  1\\
    1  & 0 &   1 &   -1\\
   1  &  0 &  1 &   1\\
\end{array}\right),
   $$

 \begin{eqnarray*} 
 && f_{T_{99}}(x,y,z)=x^4+2y^4+2z^4+3x^3y+x^3z+3xy^3+\\
&&y^3z -3z^3y+6x^2y^2+z^2y^2+8xy^2z+x^2yz-5xyz^2
 \end{eqnarray*}
 $$
T_{119}=\left( \begin{array}{cccc}  
 1 &   0 &   0 &   0  \\
  0 & 1 &  0 & 0\\
 0 &  0 &  1 &  0 \\
 0   & 0 &   0 &   1\\
\end{array}\right)
:
\left( \begin{array}{cccc}  
 0 &   0 &  0&  -1\\
   - 1 &  1 &   1 &  - 1\\
    0  & -1 &   0&    0\\
    1  &  0 &   0 &   0\\
    
 \end{array}\right)
 :
\left( \begin{array}{cccc}     
    1  & -1 &   -1 &   1\\
  1  &  -1 &   1 &  0\\
    0  & 1 &   1 &   1\\
   1  &  1 &  0 &  - 1\\
\end{array}\right),
   $$

 \begin{eqnarray*} 
 && f_{T_{119}}(x,y,z)=x^4+y^4+9z^4+x^3y+xy^3+2y^3z-2z^3y+\\
&& 2x^2y^2-3x^2z^2-z^2y^2+xy^2z+x^2yz+xyz^2.
 \end{eqnarray*}
 $$
T_{207}=\left( \begin{array}{cccc}  
 1 &   0 &   0 &   0  \\
  0 & 1 &  0 & 0\\
 0 &  0 &  1 &  0 \\
 0   & 0 &   0 &   1\\
\end{array}\right)
: 
\left( \begin{array}{cccc}  
 -1&   -1 &  -1&  1\\
   0 &  1 &   1 &  0\\
    -1  & 0 &   0&    -1\\
    -1  &  1 &   0 &   -1\\
    
 \end{array}\right)
 :   
\left( \begin{array}{cccc}    
    -1  & 1 &   -1 &   1\\
  -1  &  0 &   1 &  1\\
    -1  & 0 &   -1 &   1\\
   0  &  0 &  -1 &  - 1\\
\end{array}\right),
   $$

 \begin{eqnarray*} 
 &&f_{T_{207}}(x,y,z)=x^4+2y^4+2z^4-x^3y-3x^3z+xy^3+6y^3z\\
&&-3xz^3+2z^3y-x^2y^2+4x^2z^2+6z^2y^2+8xy^2z-3x^2yz+7xyz^2.
 \end{eqnarray*}
$$
T_{237}=\left( \begin{array}{cccc}  
 1 &   0 &   0 &   0  \\
  0 & 1 &  0 & 0\\
 0 &  0 &  1 &  0 \\
 0   & 0 &   0 &   1\\
\end{array}\right)
: 
\left( \begin{array}{cccc}  
 0&   -1 &  1& - 1\\
   0 & - 1 &  - 1 &  0\\
    0 & 1 &   -1&    0\\
    1  &  1 &   0 &   1\\
    
 \end{array}\right)
 :  
\left( \begin{array}{cccc}     
    0  & -1 &   -1 &   -1\\
  -1  &  -1 &   -1 &  0\\
    0  & -1 &   0 &   1\\
   1  &  -1 &  1 &  1\\
\end{array}\right),
   $$
\begin{eqnarray*}
&&f_{T_{237}}(x,y,z)=x^4+2y^4+3z^4-x^3y+4y^3z +3z^3y+\\
&&x^2y^2-x^2z^2+6z^2y^2+x^2yz.
\end{eqnarray*}
\section{Notations from differential geometry}
For the use in the following sections, we review some basic notations of differential geometry. For more details, for example,  see the Nomizu and Sasaki \cite{NS}. When we denote the parametrized point on the surface by $x(s,t),$ we denote its partial derivatives by  
\begin{eqnarray}
x_{s}(s,t)&=&\frac{ \partial x(s,t)}{\partial s},\\
x_{t}(s,t)&=&\frac{ \partial x(s,t)}{\partial t}.
\end{eqnarray}
\begin{definition}
\begin{equation}
E=\langle x_{s},x_{s} \rangle, \ \ F=\langle x_{s},x_{t} \rangle, \ \ G=\langle x_{t},x_{t} \rangle,
\end{equation} 
are called the first fundamental coefficients.
Putting $dx=x_{s}(s,t)ds+x_{t}(s,t)dt$, the form 
\begin{equation}
I=\langle dx,dx \rangle =Es^2+2Mdsdt+Ndt^2
\end{equation}
is called the first fundamental form of the surface.
\end{definition}
\begin{definition}
\begin{equation}
\bm{n}=\frac{x_{s} \times x_{t} } {||x_{s} \times x_{t} ||}
\end{equation}
\end{definition}
is called the unit normal vector at the point $x(s,t).$
\begin{definition}
Putting
\begin{equation}
x_{s s}=\frac{ \partial ^2 x(s,t)}{\partial s^2}, \ \ x_{s s}=\frac{ \partial ^2 x(s,t)}{\partial s t}, \ \ x_{s s}=\frac{ \partial ^2 x(s,t)}{\partial t^2}, 
\end{equation}
the scalar functions
\begin{equation}
L=\langle x_{s s },\bm{n} \rangle \ \ M=\langle x_{s t  },\bm{n} \rangle \ \  \ \ N=\langle x_{t t  },\bm{n} \rangle 
\end{equation}
are called the second fundamental coefficients. The form
\begin{equation}
II=-\langle dX,d\bm{n}\rangle=Lds^2+2Mdsdt+Ndt^2
\end{equation}
is called the second fundamental form of the surface.
\end{definition}
\begin{definition}
At the point $P$ on the surface, let $k_{1}$ and $k_{2}$ be the maximum and minimum of curvatures of  curves generated by  the intersection of the surface with the plane spanned by the normal vector and a tangent vector, 
$H=(k_{1}+k_{2})/2$ is called the mean curvature and $K=k_{1}k_{2}$ is called the Gaussian curvature.  These are calculated by 
\begin{equation}
 H=\frac{EN-2FM+GL}{2(EG-F^2)} \ \ {\rm and} \ \ K=\frac{LN-M^2}{EG-F^2}.
\end{equation}
\end{definition}
\section{Rough $SL(3)$ invariants }
For checking $GL(3)$-equivalence between two tensors $T_{1}$ and $T_{2}$, we need to test the equation  
$$
f_{T_{2}}(\bm{x})=cf_{T_{1}}(\bm{x}R),  c \in \mathbb{R} \ \  {\rm and} \ \ R \in GL(3).
$$ 
Further, $S=R/|R|^{1/3} \in SL(3),$ 
\begin{equation}
cf_{T_{1}}(\bm{x}R)= c|R|^{4/3} f_{T_{1}}(\bm{x} R/|R|^{1/3})=c'f_{T_{1}}(\bm{x}S) \  \ with \ \ c' \in \mathbb{R} \ \ and \ \  S \in SL(3).
\end{equation}
Thus, $GL(3)$-equivalence reduces to $SL(3)$-equivalence. Then, the following theorem holds.
\begin{thm}\label{SLGL}
For two tensors $T_{1}$ and $T_{2}$, assume that $f_{T_{2}}(\bm{x})=cf_{T_{1}}(\bm{x}S)$ does not hold for any choice of $c \in  \mathbb{R} \ \ {\rm and} \ \  {\rm any}  S \in SL(3)$.  
Then, $T_{1}$ and $T_{2}$ are not $GL(3)$ equivalent. 
\end{thm}  
This justifies to study $SL(3)$-equivalence among absolutely nonsingular tensors for investigating $GL(3)$ equivalence. 
\begin{remark}\label{positivec}
If $c$ is negative, then by consider $T_{3}=PT_{2}$ with $|P|<0,$  we have 
$$
f_{T_{3}}(\bm{x})=cf_{T_{1}}(\bm{x}R),  c \in \mathbb{R}^{+} \ \  {\rm and} \ \ R \in GL(3).
$$
where $T_{3}$ is equivalent to $T_{2}.$ Thus, we can assume by writing $T_{3}$ as $T_{2}$ again  
$$
f_{T_{2}}(\bm{x})=cf_{T_{1}}(\bm{x}R),  c \in \mathbb{R}^{+} \ \  {\rm and} \ \ R \in GL(3).
$$
\end{remark}
The following rough $SL(3)$ invariants are useful. 
\begin{thm}
A convex surface is transformed into a convex surface by a $SL(3)$ linear transformation and so, a tensor with a determinant polynomial whose constant surface is convex is not  equivalent to a tensor with a determinant polynomial whose constant surface is not convex.
\end{thm}
Only the tensor of No.1 has the convex surface among 8 figures in the  Figure 1.1 and so the tensor of No. 1 is not $SL(3)$ equivalent to all other tensors in the Figure 1.  
\begin{definition}
A point on the surface is called a singular point if the normal vector at the point can not be defined.
\end{definition}
\begin{thm}
If the constant surface of a tensor $T_{1}$ has a singular point and the constant surface of a tensor $T_{2}$ has no singular point, they are not  $SL(3)$($GL(3))$ equivalent. 
\end{thm}
\begin{example}
Let 
$$
T_{1}=\left( \begin{array}{cccc}
0 & 1 & 0 & 0  \\
-1 & 0 & 0 & 0 \\
0 & 0 & 0 & -1 \\
0 & 0 & 1 & 0
\end{array}
\right);
\left( \begin{array}{cccc}
0 & 0 & 1 & 0  \\
0 & 0 & 0 & 1 \\
-1 & 0 & 0 & 0 \\
0 & -1 & 0 & 0
\end{array}
\right);
\left( \begin{array}{cccc}
0 & 0 & 0 & 1  \\
0 & 0 & -1 & 0 \\
0 & 1 & 0 & 0 \\
-1 & 0 & 0 & 0
\end{array}
\right)
$$
and
\[T_{2}= \left ( \begin{array}{cccc}
0 & 1 &  0 & 0 \\
-1 & 0 & 0  & 0 \\
 0 & 0 &  1 & 0 \\
0 & 0 & 0  & 1 \\
\end{array}
\right);
 \left ( \begin{array}{cccc}
0 & 0 &  0 & -1 \\
0 & 0 & 1 & 0 \\
1 & 0 &  0 & 0 \\
0 & 1 & 0  & 0 \\
\end{array}
\right);
 \left ( \begin{array}{cccc}
0 & 0 &  1 & 0 \\
0 & 1 & 0  & 0 \\
0 & 0 &  0 & 1 \\
1 & 0 & 0  & 0 \\
\end{array}
\right),
\]
and  the determinant polynomials of $T_{1}$ and $T_{2}$ are given below respectively. Then, we have
$$
f_{T_{1}}(x,y,z)=(x^2+y^{2}+z^{2})^2,
$$
and 
$$
f_{T_{2}}(x,y,z)=(x^2+y^{2})^2+z^{4}.
$$
Note that both of them are positive definite, that is, both of $T_{1}$ and $T_{2}$ are absolutely nonsingular. It is clear that the constant surface of $f_{T_{1}}(\bm{x})$ is a sphere and it has no singular point, and on the other hand, that the constant surface of $f_{T_{2}}(\bm{x})$ is a conic and it has a singular point. Hence, $T_{1}$ and $T_{2}$ are not equivalent.
 
\end{example}
\begin{definition}
When we consider a mesh of the parameter space,it produces a lattice of points on the constant surface. Let $K_{+}$ be the number of lattice points at which the Gaussian curvature is positive, and $K_{-}$ and $K_{0}$ be defined in the same way. \\
\end{definition}
Then we have
\begin{thm}
The triplet $(K_{+},K_{-}, K_{0})$ is an $SL(3)$-invariant.
\end{thm}

\section{$SL(3)$ integral invariants}
The following Figure 2 and 3 shows the figures of convex bodies that are enclosed by constant surfaces. The number of figures corresponds to that in our list of absolutely nonsingular tensors( Maehara \cite{Maehara}). 
\begin{figure}[htbp]\label{positive}
\begin{center}
\includegraphics[width=3cm, bb= 0 0 150 150 ]{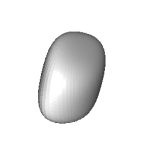}
\includegraphics[width=3cm, bb= 0 0 150 150 ]{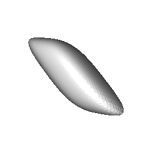}
\includegraphics[width=3cm, bb= 0 0 150 150 ]{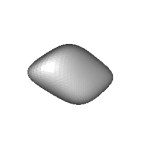}
\includegraphics[width=3cm, bb= 0 0 150 150 ]{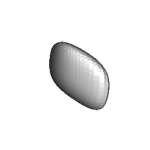}
\includegraphics[width=3cm, bb= 0 0 150 150 ]{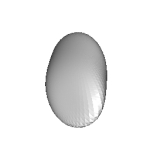}
\includegraphics[width=3cm, bb= 0 0 150 150 ]{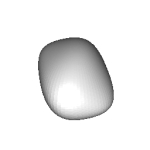}
\includegraphics[width=3cm, bb= 0 0 150 150 ]{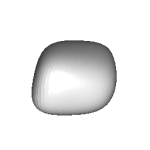}
\includegraphics[width=3cm, bb= 0 0 150 150 ]{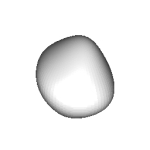}
\caption{The constant surfaces for $No1,19,22,23,42,60,61,65$ absolutely nonsingular tensors}
\end{center}
\end{figure}
\begin{figure}[htbp]\label{positivek}
\begin{center}
\includegraphics[width=3cm, bb= 0 0 150 150 ]{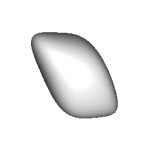}
\includegraphics[width=3cm, bb= 0 0 150 150 ]{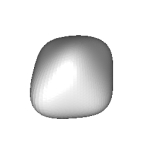}
\includegraphics[width=3cm, bb= 0 0 150 150 ]{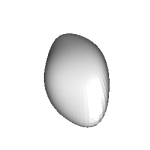}
\includegraphics[width=3cm, bb= 0 0 150 150 ]{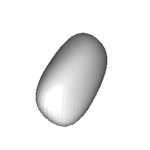}\\
\includegraphics[width=3cm, bb= 0 0 150 150 ]{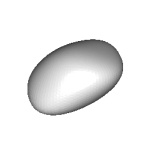}
\includegraphics[width=3cm, bb= 0 0 150 150 ]{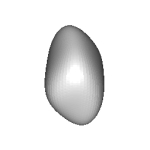}
\includegraphics[width=3cm, bb= 0 0 150 150 ]{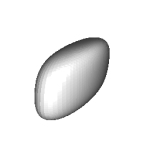}
\includegraphics[width=3cm, bb= 0 0 150 150 ]{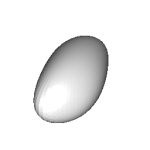}
\caption{The constant surfaces for $No 72,74,76,83,85,87,95,103$ absolutely nonsingular tensors}
\end{center}
\end{figure}
In this section, we want to find some $SL(3)$-invariant for such convex bodies. For this purpose, the following valuation theory is a quite useful.     
Here, we make a brief summary of the valuation theory from Ludwig \cite{Ludwig2} and \cite{LudwigReit}. The definition is stated for a general $p$.  
\begin{definition}
Let $\mathcal{K}$ denote the set of all convex bodies in $\mathbb{R}^{p}$. 
A functional $t(\cdot)$ from $\mathcal{K}$ to  $\mathbb{R}$ is called a valuation if it satisfies 
\begin{equation}
t(K )+t(L)=t(K \cup L ) +t(K \cap L), \ \ K, L \in \mathcal{K}.  
\end{equation}
\end{definition}
Next theorem is a starting point of characterization of invariant valuation. 
\begin{thm} \label{Hadwiger}(Hadwiger\cite{Hadwiger}).
A continuous valuation  $t(\cdot)$ from $\mathcal{K}$ to $\mathbb{R}$ is invariant with respect to rigid motion if and only if 
there are constants $c_{0},c_{1},...,c_{p}$ such that 
\begin{equation}
t(K)=c_{0}V_{0}(K)+c_{1}V_{1}(K)+\cdots+c_{p}V_{p}(K),
\end{equation}
where 
$V_{0}(K),V_{1}(K),\cdots,$ and $V_{p}(K)$ are the intrinsic volumes of $K.$ 
We remark that the volumes $V_{k}$ are called quermassintegrals of $K$ in \cite{Leicht} and that $V_{0}$
 is the Euler index, $V_{p-1}$ the affine volume of $\partial K$, and $V_{p}$ is the volume of the convex body $K.$ In the following, we simply denote by $a(K)$ as the affine volume and call it affine surface area following \cite{LudwigReit}, and denoe by $V(K)$ the volume $V_{p}(K)$.  
\end{thm}
\begin{definition}
A functional $t(\cdot)$ on $\mathcal{K}$ is said to be equi-affine invariant if it is $SL(p)$-invariant and location invariant.
\end{definition}
The following is essential for us.
\begin{thm} \label{Ludwig}(Ludwig \cite{Ludwig0} ,\cite{Ludwig1}, \cite{Ludwig2} and \cite{LudwigReit} ).
An upper semi-continuous valuation  $t(K)$ from $\mathcal{K}$ to $\mathbb{R}$ is equi-affine invariant if and only if 
there are constants $c_{0},c_{1}, \in R$ and $c_{2}\geq 0 $ such that
\begin{equation}
t(K)=c_{0}V_{0}(K)+c_{1}V_{p}(K)+c_{2}a(K),
\end{equation}
where $a(K)$ denotes the affine surface area.
\end{thm}
In short, $SL(p)$ invariant valuation is only the weighted sum of the Euler index and the volume $V_{p}(K)$and the affine surface area $a(K)$.  
This means that  
\begin{prop}\label{invariant}
$V(K)$ and $a(K)$ are $SL$-invariants.
\end{prop}
So, we adopt the volume $V(K)$ and the affine surface area $a(K)$ as indexes of $SL(p)$-equivalence. 
Further, the next proposition by Lutwak\cite{Lut}  is very useful for us.
\begin{prop}\label{homogenous}
When $p=3$,$a(K) $is homogeneous of degree $3/2$, that is, 
\begin{equation} a(dK)=d^{3/2}a(K), K \in \mathcal{K}_{0}. \end{equation} 
\end{prop}

\subsection{Volume as an $SL(3)$-invariant }
We are considering the equivalence relation among absolutely nonsingular tensors. As is shown in Theorem \ref{compact}, for such kind of tensors, the constant surfaces of them are compact. Note that from Propisition \ref{invariant} the volume of the region enclosed by the constant surface is $SL(3)$-invariant. Then, by the following Gauss's theorem, we can calculate the volume by the parametric  representation given by the equation (\ref{parametric}). 
\begin{thm}(Gaussian formula)\\
For the  region $\Omega$ enclosed by the space surface $\partial \Omega$, letting $f dy \wedge dz+ g dz \wedge dx + h dx \wedge dy$ be the differential form of 2nd degree, it holds 
\begin{equation}\label{Gauss}
\int_{ \partial \Omega} f dy \wedge dz+ g dz \wedge dx + h dx \wedge dy=\int_{ \Omega}\left(\frac{\partial f}{\partial x}+\frac{\partial g}{\partial y}+\frac{\partial h}{\partial z}\right) dx \wedge dy \wedge dz
\end{equation}
\end{thm}
We denote by $V(\Omega)$ the volume of the region $\Omega$. By this formula, we have
\begin{equation}
V(\Omega)=LHS \ \ of \ \ the \ \ equation (\ref{Gauss}).  
\end{equation}
For the present case, by use of the spherical coordinates (s,t), the point of the boundary $\partial \Omega$ is parametrized as $\bm{x}=r(s,t)(\Phi_{x}(s,t),\Phi_{y}(s,t),\Phi_{z}(s,t)).$ Hence, we have
\begin{eqnarray}
dy &=& \frac{\partial y}{\partial s}ds+\frac{\partial y}{\partial t}dt\\ \nonumber 
& =& \left((-1/4)\frac{dp/ds}{p^{5/4}}\Phi_{y}+\frac{d\Phi_{y}/ds }{p^{1/4}}\right) ds +\\ \nonumber 
&& \left ((-1/4)\frac{dp/dt}{p^{5/4}}\Phi_{y}+\frac{\Phi_{y}/dt}{p^{1/4}}\right) dt \\ \nonumber 
dz &=& \frac{\partial z}{\partial s}ds+\frac{\partial z}{\partial t}dt, \\ \nonumber 
&=&\left((-1/4)\frac{dp/ds}{p^{5/4}}\Phi_{z}+\frac{d\Phi_{z}/ds }{p^{1/4}}\right) ds +\\ \nonumber 
& & \left ((-1/4)\frac{dp/dt}{p^{5/4}}\Phi_{z}+\frac{\Phi_{z}/dt}{p^{1/4}}\right) dt. \\
\end{eqnarray}
Therefore,
\begin{eqnarray*}
dy \wedge dz& =&  \left((-1/4)\frac{dp/ds}{p^{5/4}}\Phi_{y}+\frac{d\Phi_{y}/ds }{p^{1/4}}\right) \left ((-1/4)\frac{dp/dt}{p^{5/4}}\Phi_{z}+\frac{\Phi_{z}/dt}{p^{1/4}} \right)ds \wedge dt  \\ \nonumber 
&& -\left ((-1/4)\frac{dp/dt}{p^{5/4}}\Phi_{y}+\frac{\Phi_{y}/dt}{p^{1/4}}\right)  
\left((-1/4)\frac{dp/ds}{p^{5/4}}\Phi_{z}+\frac{d\Phi_{z}/ds }{p^{1/4}} \right) ds \wedge dt. 
\end{eqnarray*}
Similarly, 
\begin{eqnarray*}
 dz \wedge dx&=&  \left((-1/4)\frac{dp/ds}{p^{5/4}}\Phi_{z}+\frac{d\Phi_{z}/ds }{p^{1/4}}\right) \left ((-1/4)\frac{dp/dt}{p^{5/4}}\Phi_{x}+\frac{\Phi_{x}/dt}{p^{1/4}}\right)ds \wedge dt  \\ \nonumber 
&& -\left ((-1/4)\frac{dp/dt}{p^{5/4}}\Phi_{z}+\frac{\Phi_{z}/dt}{p^{1/4}}\right) 
\left((-1/4)\frac{dp/ds}{p^{5/4}}\Phi_{x}+\frac{d\Phi_{x}/ds }{p^{1/4}}\right)  ds \wedge dt,
\end{eqnarray*}
and
\begin{eqnarray*}
dx \wedge dy&=&\left((-1/4)\frac{dp/ds}{p^{5/4}}\Phi_{x}+\frac{d\Phi_{x}/ds }{p^{1/4}}\right) \left ((-1/4)\frac{dp/dt}{p^{5/4}}\Phi_{y}+\frac{\Phi_{y}/dt}{p^{1/4}}\right)ds \wedge dt  \\ \nonumber 
&& -\left ((-1/4)\frac{dp/dt}{p^{5/4}}\Phi_{x}+\frac{\Phi_{x}/dt}{p^{1/4}}\right) 
\left((-1/4)\frac{dp/ds}{p^{5/4}}\Phi_{y}+\frac{d\Phi_{y}/ds }{p^{1/4}}\right)  ds \wedge dt.
\end{eqnarray*}
By using these,  we can calculate the volume of the region enclosed by the constant surface of the determinant polynomial. Let $T_{1}$ and $T_{2}$ be two $ n \times n \times 3 $ tensors and let $\Omega_{i}$ denote the regions $\{\\bm{x}|f_{T_{i}}(\bm{x}) \leq 1\}$, which are enclosed by the surfaces of $\{\bm{x}| f_{T_{1}}(\bm{x}) = 1\}$.  
Then, by $SL(3)$ invariance of volumes, we have 
\begin{thm} \label{SLequivalentthm}
If $V_{1} \neq V_{2},$  $f_{T_{2}}(\bm{x})  \neq f_{T_{1}}(\bm{x}R)$ for any $R \in SL(3)$, namely,  $T_{1}$ and $T_{2}$ are not $SL(n) \times SL(n) \times SL(3)$ equivalent.    
\end{thm}
For $GL$ invariance, the next lemma is helpful.
\begin{lemma}\label{dvolume}
For a determinant polynomial $f(\bm{x}),$ let $V(c)$ be  the volume of $\Omega_{c}=\{\bm{x}|cf(\bm{x})\leq 1\}.$ Then $V(c)=c^{-3/4}V(1)$ for $4 \times 4 \times 3$ case.
\end{lemma}
\begin{proof}
By changing a polynomial $f(\bm{x})$ into its constant multiple $cf(\bm{x})$, the coordinates  $\bm{x}(s,t)$ on the constant surface are subject to changes 
to $(\frac{1}{c})^{1/4}\bm{x}(s,t)$. Hence, the integral 
\begin{equation}
\frac{1}{3}\int_{ \partial \Omega} x dy \wedge dz+ y dz \wedge dx + z dx \wedge dy
\end{equation}
is multiplied by $c^{-3/4}.$ This proves the assertion of Lemma \ref{dvolume}.
\end{proof}
\begin{thm}\label{volume}
Assume that  $T_{1}$ and $T_{2}$ be $SL(3)$ equivalent and therefore that there is a relation between their determinant polynomials, 
\begin{equation} \label{SLequivalent}
f_{T_{2}}(\bm{x})=cf_{T_{1}}(\bm{x}R), 
\end{equation}  
where $c \in \mathbb{R}$ and $R \in SL(3)$.  Let $V_{1}(c)$ and $V_{2}(c)$ denote the volumes of $\Omega^{(1)}_{c}=\{\bm{x}|cf_{T_{1}}(\bm{x})\leq 1\}$ and  $\Omega^{(2)}_{c}=\{\bm{x}|cf_{T_{2}}(\bm{x})\leq 1\}$ respectively. Then, it holds that
\begin{equation}
c=(V_{1}/V_{2})^{4/3}.
\end{equation}

\end{thm}
\begin{proof}
The proof is trivial from Lemma \ref{dvolume} and omitted.
\end{proof}
From Theorem \ref{volume}, we can know the constant $c$ in the equation (\ref{SLequivalent}). In the next section, it will be made clear that 
this expresson is helpful for establishing $GL(3)$-equivalence.  

\subsection{Affine surface area as an $SL(3)$-invariant}
 In this section, for testing $SL(3)$-equivalence, we propose to use the affine surface area, which is an $SL(3)$-invariant by Theorem \ref{invariant}.
 When $p=3$, the affine surface area has the following integral expression.  
\begin{definition}
For a smooth convex body $K \subset \mathbb{R}^{3}$, the affine surface area is given by 
\begin{equation}
a(K)=\int_{\partial K}\kappa(K,\bm{x})^{\frac{1}{4}}\sqrt{EG-F^2}dsdt,
\end{equation}
where  $\kappa(\partial K,\bm{x})$ is the Gaussian curvature and  $E,F$ and $G$ denote the first fundamental coefficients.
\end{definition}
Next, we show that the affine surface area is useful even as a tester of $GL(3)$-equivalence.  Assume that  we know the constant $c$ in the relation $f_{T_{2}}(\bm{x})=cf_{T_{1}}(\bm{x}R)$ with $c \in \mathbb{R}^{+}$ and $R \in SL(3)$ by Theorem \ref{volume}. Then, 
\begin{eqnarray}
 \Omega_{2}&=&\{\bm{x}| f_{T_{2}}(\bm{x}) \leq1 \}\\ \nonumber
&=& \{ \bm{x}|c f_{T_{1}}(\bm{x}) \leq 1 \}\\ \nonumber
&=&   \{ \bm{x}|f_{T_{1}}(c^{1/4}\bm{x}) \leq 1 \}\\ \nonumber
&=&   c^{-1/4} \{\bm{x}| f_{T_{1}}(\bm{x}) \leq 1 \}\\ \nonumber
&=&   c^{-1/4}  \Omega_{1}.\\
\end{eqnarray}
From Proposition \ref{homogenous}, we have
\begin{equation}
a(\Omega_{2})=c^{-3/8} a(\Omega_{1}).
\end{equation}
Thus, we have the following.
\begin{thm} \label{SLGL1}
Let $T_{1}$ and $T_{2}$ be absolutely nonsingular tensors. Noting Remark \ref{positivec}, by Theorem \ref{volume}, we can obtain the estimate of $c \in R^{+}$ under the assumption that their determinant polynomials have the relation $f_{T_{2}}(\bm{x})=cf_{T_{1}}(\bm{x}R)$ for some unknown constant $c \in \mathbb{R}^{+}$ and an unknown matrix $R \in SL(3).$ 
Then, if  $a(\Omega_{2}) \neq c^{-1/8} a( \Omega_{1})$, $T_{1}$ and $T_{2}$ are not $GL(3)$ equivalent. 
\end{thm}
 By using Theorem \ref{dvolume}, this is rephrased as 
 \begin{thm}\label{SLGL2}
 Let $T_{1}$ and $T_{2}$ be absolutely nonsingular tensors. 
Then, if 
 \begin{equation} 
a(\Omega_{2}) \neq \left( \frac{V(\Omega_{2})}{V(\Omega_{1})} \right) ^{1/2} a(\Omega_{1}), 
\end{equation}
$T_{1}$ and $T_{2}$ are not $GL(3)$-equivalent, where $V(\Omega_{1})$ and $V(\Omega_{2})$ denote the volume of $\Omega_{1}$ and $\Omega_{2}$ respectively.
 \end{thm}
\section{Integral $GL(3)$-invariant}
In the latter half of the previous section, we presented a procedure to test a non-$GL(3)$-equivalence, however, it is somewhat indirect because we need to estimate the constant $c$ before starting the procedure.
In this section, we consider a direct method handling non-equivalence by using a generalized affine surface area.
That is, we  consider the $L_{q}$ affine surface area, which is an extension of the affine surface area and developed by Letwak \cite{Lut}.  Hug \cite{Hug} gave an equivalent definition. The following is the Hug' s definition. 
\begin{definition}
\begin{equation}
L_{q}a(K)= \int_{\partial K} \kappa_{0}(K,\bm{x})^{ \frac{q}{p+q}} d\sigma_{K}(\bm{x})  
\end{equation}
where
\begin{equation}
\kappa_{0}(K,\bm{x})=\frac{\kappa(K,\bm{x})}{ \langle \bm{x},\bm{n}(K,\bm{x})\rangle^{p+1}}.   
\end{equation}
and $d\sigma_{K}(\bm{x})$ is called a cone measure defined by 
\begin{equation}
d\sigma_{K}(\bm{x})=\langle \bm{x},\bm{n}(K,\bm{x})\rangle d\bm{x}, 
\end{equation}
and  $\bm{n}(K,\bm{x})$ denotes the outer normal at $\bm{x}$ on $\partial K.$
\end{definition}
When $q=1$, $L_{q}(K)$ becomes the affine surface area $a(K)$, and when $q=p$,  it becomes a classical centro-affine surface area $a_{c}(K)$ thta is defined as 
\begin{equation}
a_{c}(K)=\int _{\partial K}\kappa_{0}(K,\bm{x})^{1/2}d\sigma_{K}(\bm{x}),
\end{equation}
which is known to be  $GL(p)$-invariant. The characterization of a general $GL(p)$-invariant functional is given  below.
\begin{thm}({\rm Ludwig and Reitzsner \cite{LudwigReit} })
Let $\mathcal{K}_{0}$ be the space of convex bodies that contain the origin in their interiors.  
An upper semi-continuous functional  $t(\cdot)$ from $\mathcal{K}_{0}$ to $\mathbb{R}^{1}$ is $GL(p)$-invariant  if and only if 
there are nonengative constants $c_{0}$ and $c_{1}$ such that 
\begin{equation}
t(K)=c_{0}V_{0}(K)+c_{1}a_{c}(K).
\end{equation}
\end{thm}

\section{Spherical design}
According to our experiments, the numerical integrations of the invariants must be accurate at least 2 decimals. So, the caluculations of the invariants are a little bit heavy. In this section, we consider the t-design method as an substitute of the nemerical integrations. The spherical design was initiated by Delsarte et al. \cite{DGS}  and has been studied by several researchers, for example, see Bannai and Bannai \cite{BB}. It is defined as follows.  
\subsection{An overview of spherical design}
\begin{definition}
A finite set $X$ on the sphere is called t-spherical design if the following equality holds that for any polynomial $f(x,y,z)$ with a degree less than or equal to $t,$ 
\begin{eqnarray}\label{tdesignintegral}
\frac{1}{|S^{2}|}\int_{S^{2}}f(x,y,z)d\sigma=\frac{1}{|X|}\sum_{(x,y,z) \in X} f(x,y,z),
\end{eqnarray}
where $S^{2}$ denotes the unit sphere of $\mathbb{R}^{3}$ and $d\sigma$denotes the surface element of the sphere and  $|S^{2}|$ denotes the surface  area of the sphere.
\end{definition}
A parametrized integral formula of the equation \ref{tdesignintegral}  is given by 
\begin{equation}
\int f(s,t) \sin(s) dsdt =\frac{4\pi}{N}
\sum_{i=1}^{N}f(s_{i},t_{i}),
\end{equation}
where $(s_{i},t_{i}),i=1,2,....,N$  are the corresponding parameters to the design points in $X$. One point  to overcome for our purpose is that we need to integrate  some nonlinear functions that are not polynomials and hence we can not use any t-design directly. However, we can rely on  the next theorem to solve this point.
\begin{thm}
Let $f(x,y,z)$ be an continuous function over the unit sphere and let $\epsilon_{t}$ be a positive number  such that
$$|f(x,y,z)-p(x,y,z)|<\epsilon_{t}$$ 
uniformly for some polynomial $p(x,y,z)$ with degree less than or equal to $t$. 
Then, it holds that 
\begin{equation}
|\int_{\partial S} f(x,y,z) dS - 4\pi\frac{1}{N}\sum_{i=1}^{N}f(x_{i})|<8\pi \epsilon_{t}
\end{equation}
\end{thm}
\begin{proof}
\begin{eqnarray}
&&\left| \int_{\partial S} f(x,y,z) dS - 4\pi\frac{1}{N}\sum_{i=1}^{N}f(x_{i}) \right| \\ \nonumber 
&\leq& \left| \int_{\partial S} f(x,y,z) dS - 4\pi\frac{1}{N}\sum_{i=1}^{N}p(x_{i})\right|+\left| 4\pi\frac{1}{N}\sum_{i=1}^{N}p(x_{i})- 4\pi\frac{1}{N}\sum_{i=1}^{N}f(x_{i}) \right|\\ \nonumber 
&=&\left| \int_{\partial S} f(x,y,z) dS - \int_{|\partial S}  p(x_{i})dS \right|+\frac{4\pi}{N}  \left| \sum_{i=1}^{N}|p(x_{i})- f(x_{i})|\right| \\ \nonumber
&\leq &    \int \epsilon_{t}  dS+ 4\pi \epsilon_{t}\\ \nonumber
&=& 8 \pi \epsilon_{t}  
\end{eqnarray}
\end{proof}
\begin{remark}
By the above theorem, we need not to know the best approximate polynomial concretely in order to obtain an approximate value of  the  integration,  and  it is enough to use  $f(x,y,z)$ itself. Moreover the error of the approximation is bounded from above by the  multiple of $\epsilon_{t}$ by $8\pi.$ 
For a substantial evaluation of the approximation, we need to know $\epsilon_{t}$. The problem is interesting, however, it is a little bit heavy task at present, and so it is postponed to the future work.  
\end{remark}
\subsection{Calculation of integral invariants by a 20-design }
Using the result of the previous subsection, we consider the integration
\begin{equation}
a(K)=\int \kappa(s,t)^{1/4}\sqrt{EG-F^2}dsdt.
\end{equation}
where $s,t$ moves $0<s<\pi$,$0<t<2\pi$.   This integration can be thought to be an integration over the unit sphere by 
\begin{eqnarray}
a(K)&=&\int_{(s,t,) \in [0,\pi]\times[0,2\pi]}   \kappa(s,t)^{1/4}\sqrt{EG-F^2}dsdt\\ \nonumber
           &=&\int_{(s,t,) \in [0,\pi]\times[0,2\pi]}  \kappa(s,t)^{1/4}\frac{\sqrt{EG-F^2}}{\sin s}\sin s dsdt\\ \nonumber
           &=&\int_{\partial S} \kappa(s,t)^{1/4}\frac{\sqrt{EG-F^2}}{\sin s}dS,\\
\end{eqnarray}
where $dS=sin(s)$. Hence,
\begin{equation}
p(x,y,z)=\kappa(s,t)^{1/4}\frac{\sqrt{EG-F^2}}{\sin s}
\end{equation}
is taken to be  a function over the unit sphere and so the integral invariant can be approximated by the right hand side of the equation below.  
\begin{equation} \label{approximation}
\int_{\partial S} \kappa(s,t)^{1/4}\frac{\sqrt{EG-F^2}}{\sin s}dS \sim  \frac{4\pi}{N}\sum_{i=1}^{N}p(x_{i},y_{i},z_{i})
\end{equation}
 The values of invariants  calculated by the lattice method and the 20-design method will be give in the next section. The 20-design method show very nice approximations in some cases, however, do not show good approximations for other cases. That is, for our integration of invariants, the spherical design method does not give stable values, unfortunately. This might suggest that we need to use design with more higher degree than 20.  
 \section{Effectiveness of the invariants as testers of non-equivalence }
 In this section, we will show the effectiveness of the numerical values of the invariants as testers of non-equivalence. We numerically calculated the volume $V(\Omega)$, the affine surface area $a(\Omega)$and centro-affine surface area $a_{c}(\Omega)$ of the region $\Omega=\{\bm{x}|f(\bm{x}\leq 1\}$ defined by the determinant polynomials $f_{T}(\bm{x})$. As examples, we calculate them for the 16 tensors which are in $\mathcal{K}_{0}$, whose constant surfaces are figured in Figures 2 and 3 in the section 5. The numerical calculations are performed in two way, that is, by the lattice method and by the t-design method, and they are compared. As for the t-design method, we use the 20-design named des.3.216.20 in \cite{Slone} which has 216 points. In the tables below, M1-P2-G5,  M6-P2-G, M1-P2-G7 and 20-design denote the globally adaptive integration with accuracy of 5 digits, pseudo-Monte Carlo integration,  the globally adaptive integration with accuracy of 7 digits by 64 decimal calculation and 20-design method by IEEE754 decimal calculation,  respectively. For all calculation were done by Mathematica. 
\begin{table}[!htbp]\label{SLinvariant1}
\begin{center}
\begin{tabular}{|c|c|c|c|c|} \hline
 Tensor       &        V0     &    V1  & V2  & V3         \\ \hline \hline
T001   & 2.9197794095194  & 2.9197794099529 & 2.9197794089308  &   2.9197794061274  \\ \hline
T019   & 4.0314824331814 &  4.0314824340674  & 4.0314824332515 &  4.0314824319603 \\ \hline
T022   & 3.6306602017309  & 3.6306602004447  &3.6306602054741 &3.6306602016552 \\ \hline
T023   &    3.4355628950802 &3.4355628819358 &3.4355628878857& 3.4355628897838 \\ \hline
T042   &    3.7515624235646 & 3.7515624142272  & 3.7515624197586&3.7515624152774  \\ \hline
T060    &    2.1440485535226  & 2.1440485507771& 2.1440485551454&2.1440485550215 \\ \hline
T061    &    2.8594583429857 & 2.8594583441125 &2.8594583445567 & 2.8594583445567 \\ \hline
T065   &  3.1084258968340 & 3.1084258946417 &3.1084258957994&3.1084258984271  \\ \hline
T072  &    4.6861403575076 &4.6861403597489 & 4.6861403542060 &4.6861403560079\\ \hline 
T074 &   3.6302252513670 &3.6302253269919&3.63022533280632 & 3.6302253350968 \\ \hline
\end{tabular}
\end{center}
\caption{Volumes by M1-P2-G7 :$Tn$, where $n=001, 0019, 022, 023, 042,060,061$, $065, 072$ and $074$ Each line denoted as $T$n-0 lists the value of the original tensor and the lines $T$n-i,$i=1,2,...,5$ list the values for the transformed tensors of Tn by a randomly chosen matrix of $SL(3)$. }
\end{table}
\begin{table}[!htbp]\label{SLinvariant12}
\begin{center}
\begin{tabular}{|c|c|c|c|c|} \hline
 Tensor &      M1-P2-G5   &   M6-P2-G5   &   M1-P2-G7  &   20-design \\ \hline \hline
 
T001-0  &  9.961493457 &   9.962796404  &  9.961471493  &  9.90317 \\ \hline
T001-1  &   9.961470358 &   9.961249135 &   9.961471489  &  8.73023 \\ \hline
T001-2  & 9.961471133  &  9.971266750  & 9.961471486  &  9.96328 \\ \hline
T001-3  &  9.961470327  &  9.959509709  &  9.961471474 &   9.79057 \\ \hline
T001-4  & 9.961471186   & 9.979456186  &  9.961471478  & 10.88367 \\ \hline
T001-5  &  9.961474220 &   9.997180989   & 9.961471490  & 10.99278 \\ \hline \hline

T019-0 &11.560007113 & 11.560055546 & 11.560007991 & 11.87277\\ \hline
T019-1 &11.560007742 & 11.552017302  &11.560007993 & 11.69239\\ \hline
T019-2 &11.560008558  &11.559866424 & 11.560007993 & 11.51344\\ \hline
T019-3 &11.560007692  &11.501609971 & 11.560007989  &13.36769\\ \hline
T019-4 &11.560008494 & 11.545260017 & 11.560007991 & 10.40203\\ \hline
T019-5 &11.560001924 & 11.558176522  &11.560007996 & 11.49393\\ \hline \hline

T022-0 &11.020675684&  11.024551464 & 11.020674135&  11.05202\\ \hline
T022-1 &11.020673831 & 11.016947424 & 11.020674138  &11.45345\\ \hline
T022-2& 11.020676195 & 11.027525350 & 11.020674147 & 11.54386\\ \hline
T022-3 &11.020673214 & 11.016006939&  11.020674140 & 11.07524\\ \hline
T022-4 &11.020675399 & 11.031596952&  11.020674133 & 11.37096\\ \hline
T022-5 &11.020674431 & 11.022431931 & 11.020674135 & 10.74089\\ \hline \hline

T023-0 &10.771760482 & 10.773095422  & 10.771760351  & 10.73293\\ \hline
T023-1 &10.771758801 &10.774759865   &10.771760349   & 9.26881\\ \hline
T023-2 &10.771759301  &10.725843291   &10.771760352  & 10.94587\\ \hline
T023-3 &10.771759730 & 10.766135806  & 10.771760352  & 13.23848\\ \hline
T023-4 &10.771757516 & 10.773059224  & 10.771760350 &  10.78732\\ \hline
T023-5 &10.771759533 & 10.773749461   &10.771760351   &10.88149\\ \hline \hline

T042-0  &11.136697741  & 11.136755128  & 11.136697332  & 10.99424\\ \hline
T042-1  &11.136695637 &  11.140565725  & 11.136697314 &  11.28015\\ \hline
T042-2  &11.136699257  & 11.140835239  & 11.136697323  & 11.89052\\ \hline
T042-3  &11.136721676  & 11.203272583   &11.136697308   &12.13058\\ \hline
T042-4  &11.136696147 &  11.106329119   &11.136697270  &  9.81007\\ \hline
T042-5  &11.136697731 &  11.107102150  & 11.136697313  & 13.99432\\ \hline \hline
\end{tabular}
\end{center}
\caption{Affine surface area:$Tn$, where $n=001, 019, 022, 023$ and $042$ Each line denoted as $T$n-0 lists the value of the original tensor and the lines $T$n-i,$i=1,2,...,5$ list the values for the transformed tensors of Tn by a randomly chosen matrix of $SL(3)$.  }
\end{table}
\begin{table}[!htbp]\label{SLinvariant2}
\begin{center}
\begin{tabular}{|c|c|c|c|c|} \hline
Tensor &      M1-P2-G5   &   M6-P2-G5   &   M1-P2-G7  &   20-design \\ \hline \hline
T060-0  & 8.704587985  &  8.705126156  &  8.704588101 &   8.74300\\ \hline
T060-1  & 8.704590255   & 8.781085267   & 8.704588109  &  8.50058\\ \hline
T060-2  & 8.704596276   & 8.705498658  &  8.704588101  &  8.73210\\ \hline
T060-3 &  8.704588380  &  8.711001740 &   8.704588104 &   8.80910\\ \hline
T060-4  & 8.704586669  &  8.703029024  &  8.704587984 &   8.85973\\ \hline
T060-5  & 8.704588143  &  8.705901809  &  8.704588100  &  8.56701\\ \hline \hline
T061-0 & 9.759043314 &  9.759635865 &  9.759045706 &  9.741275\\ \hline
T061-1 & 9.759036154 &  9.759076500  & 9.759045704 &  9.72403\\ \hline
T061-2&  9.759050068 &  9.748685352 &  9.759045707&   9.56041\\ \hline
T061-3 & 9.759044653  & 9.734392909  & 9.759045710 &  9.76040\\ \hline
T061-4 & 9.759058206  & 9.745677922  & 9.759045694 & 10.37056\\ \hline
T061-5&  9.759046974  & 9.755984062   &9.759045716  & 8.72501\\ \hline \hline

T065-0 &10.273389075 & 10.274251947 & 10.273389369 & 10.33927\\ \hline
T065-1 &10.273387633  &10.260497042 & 10.273389360 & 10.76367\\ \hline
T065-2 &10.273389342 & 10.277789370 & 10.273389368 & 10.26620\\ \hline
T065-3 &10.273388029&  10.249526640&  10.273389366&  10.42661\\ \hline
T065-4& 10.273389939 & 10.276245030 & 10.273389370 & 10.06295\\ \hline
T065-5& 10.273397527 & 10.279052599  &10.273389365 & 10.33636\\ \hline \hline

T072-0 &12.483701912&  12.483843205 & 12.483691274&  12.67586\\ \hline 
T072-1 &12.483689616 & 12.484388161 & 12.483691282 & 12.38965\\ \hline
T072-2& 12.483690234 & 12.481034408&  12.483691282 & 10.30731\\ \hline
T072-3 &12.483690116 & 12.498107747 & 12.483691264 & 11.96858\\ \hline
T072-4 &12.483698348 & 12.435166726 & 12.483691276 & 11.219584\\ \hline
T072-5 &12.483686195 & 12.508162438&  12.483691276 & 10.183837\\ \hline \hline

T074-0& 10.732327625&  10.732623087 & 10.732332110 & 10.80078\\ \hline
T074-1 &10.732332283 & 10.726775221 & 10.732332117  &10.55820\\ \hline
T074-2 &10.732337739 & 10.734008328&  10.732332113&  11.20578\\ \hline
T074-3& 10.732332889 & 10.724453313 & 10.732332112  &10.87848\\ \hline
T074-4& 10.732331733 & 10.729148909 & 10.732332214 & 10.95073\\ \hline
T074-5 &10.732329467 & 10.727707310 & 10.732332111 & 10.57017\\ \hline \hline 

\end{tabular}
\end{center}
\caption{Affine surface area:$Tn$, where $n=060, 061, 065, 072$ and $074$ Each line denoted as $T$n-0 lists the value of the original tensor and the lines $T$n-i,$i=1,2,...,5$ list the values for the transformed tensors of Tn by a randomly chosen matrix of $SL(3)$. }
\end{table}
Table 1 shows that the SL invariance of volumes of the redions enclosed by the constant surface is clealry seen numerically for every absolutely nonsingular chosen tensors.  Tables 2 and 3 of the affine surface area show that the affine suraface area is SL invariant and that all relevant tensors are not $SL(4) \times SL(4) \times SL(3)$ equivalent mutually. From Theorem \ref{SLGL2},  combining the volume data, 
we also conclude that they are  not $GL$ equivalent. This last fact is also derived by a direct usage of the centro-affine surface data which is seen in Table 4 and Table 5. 
\begin{table}[!htbp]\label{GLinvariant1}
\begin{center}
\begin{tabular}{|c|c|c|c|c|} \hline
 Tensor &    M1-P2-G7 \ &         M6-P2-G5     &       M1-P2-G5  &   20-design        \\ \hline \hline

T001-0  & 11.690150892617500  & 11.687899476332365  & 11.687898336789288 &  11.59968  \\ \hline
T001-1  &  11.751922920525157  & 11.687898955213611 &  11.687898343722015 &   8.421025\\ \hline
T001-2  &11.689694693901319  & 11.687898370365255   &11.687898355824357 & 11.68469   \\ \hline
T001-3 &11.721355953709315   &11.687894829195568   &11.687898343333765  & 11.29880   \\ \hline
T001-4  &11.692877418652227  & 11.687897242831659   &11.687897875631138 & 10.59083   \\ \hline
T001-5  &11.679997753276740  & 11.687897167430656   &11.687898359334835& 11.46900   \\ \hline \hline

T019-0  &11.509733354093680  & 11.509334536488204   &11.509333804897551 & 11.81248  \\ \hline
T019-1   &11.472821548199051   &11.509332873485376   &11.509333807975230& 12.65290  \\ \hline
T019-2    &11.509963231209824  & 11.509337447098934   &11.509333799194381&11.32764  \\ \hline
T019-3   &11.552527050941017  & 11.509334159193976   &11.509333800434498& 11.38444  \\ \hline
T019-4   &11.495864684062547 &  11.509335287391230   &11.509333801132503  & 11.37034   \\ \hline
T019-5 &11.522546264759133  & 11.509333512497239  & 11.509333798526679  & 22.50564   \\ \hline \hline

T022-0    & 11.574282949497377  & 11.568790730790213  & 11.568790156808308&11.59771 \\ \hline
T022-1   & 11.570887655990263   &11.568785645774059  & 11.568790144452251& 11.63356 \\ \hline
T022-2  & 11.568787229271756  & 11.568788230429048   &11.568790347476747 & 11.90185 \\ \hline
T022-3   & 11.567926875696718  & 11.568789765657722  & 11.568790134358534 &11.57677 \\ \hline
T022-4  & 11.597381830882211 &  11.568789372237021  & 11.568790132199215 & 13.39389 \\ \hline
T022-5   &11.561310960443544   &11.568790290463560  & 11.568790451975676  &11.44911 \\ \hline \hline

T023-0  &  11.631078897689606  & 11.626439742081966  & 11.626439153758515 & 11.57877\\ \hline
T023-1     &11.619997976153462  & 11.626439518693340   &11.626439146934238 &11.39835\\ \hline
T023-2   &11.611266008293132   &11.626440231360507  & 11.626439154231153 &11.43501  \\ \hline
T023-3   &11.647477652583963   &11.626433368429471 &  11.626439151521914 &13.20628 \\ \hline
T023-4 & 11.607565993095791  & 11.626439544654837  & 11.626439154062155  &10.72795  \\ \hline
T023-5   & 11.620421522261536  & 11.626439658214246   &11.626439152653352 &11.74869 \\ \hline \hline

T042-0   &11.502624357421948  & 11.504755263366923  & 11.504752079092657& 11.30545  \\ \hline
T042-1  &  11.507105268508006  & 11.504753311150279  & 11.504752086650519& 11.07124 \\ \hline
T042-2  &11.519424921951189   &11.504753899401661   &11.504752079220612&   9.41924 \\ \hline
T042-3   &11.501095106227783   &11.504764044950799   &11.504752085646500& 12.16150  \\ \hline
T042-4   &  11.530791206130419  & 11.504752412140897  & 11.504752073761618& 10.11515\\ \hline
T042-5  & 11.499503647742464   &11.504752956532382   &11.504752076792938 &10.64605  \\ \hline \hline

\end{tabular}
\end{center}
\caption{Centro-affine surface area:$Tn$, where $n=001, 019, 022, 023$ and $042$ Each line denoted as $T$n-0 lists the value of the original tensor and the lines $T$n-i,$i=1,2,...,5$ list the values for the transformed tensors of Tn by a randomly chosen matrix of $GL(3)$. }
\end{table}

\begin{table}[!htbp]\label{GLinvariant2}
\begin{center}
\begin{tabular}{|c|c|c|c|c|} \hline
 Tensor    &        M6-P2-G5    &        M1-P2-G5    &        M1-P2-G7 &     20-design\\ \hline \hline
T060-0  & 11.989971644000403  & 11.989476119401702  & 11.989477685702977 & 12.05017 \\ \hline
T060-1 & 11.990418566344240 &  11.989479611584825  & 11.989477723348062&  12.01483  \\ \hline
T060-2 & 11.976852295964006  & 11.989478648296963   &11.989477738864300&  11.80414  \\ \hline
T060-3    & 11.994778478025750 &  11.989477993940143   &11.989477740049253& 11.64148\\ \hline
T060-4 &  11.989039776987073  & 11.989478217916079   &11.989477724638414 & 12.09719\\ \hline
T060-5  & 12.046738095770048   &11.989477160425962  & 11.989477721024015 & 12.11083  \\ \hline \hline

T061-0   & 11.519673795399891 &  11.518135424201142  & 11.518135117247486 & 11.48023\\ \hline
T061-1    &11.518661618867961  & 11.518134427948641  & 11.518135113069415& 11.45852 \\ \hline
T061-2     & 11.519323023325257  & 11.518135433182912   &11.518135109619174&11.45517\\ \hline
T061-3   & 11.517878708284445  & 11.518130456996306 &  11.518135109891727 & 11.66299\\ \hline
T061-4   & 11.517606307852915 &  11.518129721219993 &  11.518135107795112&11.47192  \\ \hline
T061-5  &  11.531441863084249   &11.518134741001642  & 11.518135110921667& 10.31802 \\ \hline \hline

T065-0   &  11.660951650838694  & 11.660077139783777  & 11.660146606151409& 11.77330\\ \hline
T065-1   & 11.657097464096733   &11.660146835636129  & 11.660146583158155& 11.64542 \\ \hline
T065-2     & 11.662671583310311  & 11.660135616613492   &11.660146602669996&11.69570\\ \hline
T065-3      &11.657074111599187  & 11.660148534680922  & 11.660146593520388&11.58150\\ \hline
T065-4   &11.661605706427124  & 11.660146860738219  & 11.660146601870989& 12.45337  \\ \hline
T065-5    & 11.668831112582931   &11.660147254608734   &11.660146596730511 &11.76209\\ \hline \hline

T072-0 & 11.545589518947570  & 11.545142097769179  & 11.545141226929544 & 11.76647  \\ \hline
T072-1  &11.545716622630585 &  11.545139592226001  & 11.545141221483210  & 11.52675 \\ \hline
T072-2  & 11.562200209044268  & 11.545142319259361  & 11.545141224116661&  8.50885  \\ \hline
T072-3  &  11.575704218963165  & 11.545144648175810  & 11.545141226744587& 10.30769 \\ \hline
T072-4    & 11.535076590703719  & 11.545140326506765  & 11.545141235723655&12.22675 \\ \hline
T072-5   &11.545910129113601   &11.545141536682674   &11.545141208615009 &11.19618  \\ \hline \hline

T074-0    & 11.116314213623787  & 11.116088526600165  & 11.116090556639371& 11.20632\\ \hline
T074-1  & 11.109183432623630   &11.116086365050504  & 11.116090553382580 &11.09112  \\ \hline
T074-2 &  11.121063605466493 &  11.116094135711593  & 11.116090554260284 &  9.58706 \\ \hline
T074-3   & 11.090134159234779  & 11.116089713933696   &11.116090550551305& 10.19594 \\ \hline
T074-4    & 11.135697898280837  & 11.116091869446610  & 11.116090554811293& 11.15140\\ \hline
T074-5   &  11.117481923868303  & 11.116094503240262 &  11.116090554606608&11.08749 \\ \hline \hline
\end{tabular}
\end{center}
\caption{Centro-affine surface area:$Tn$, where $n=060, 061, 065, 072$ and $074$ Each line denoted as $T$n-0 lists the value of the original tensor and the lines $T$n-i,$i=1,2,...,5$ list the values for the transformed tensors of Tn by a randomly chosen matrix of $GL(3)$. }
\end{table}
Indeed, Tables 4 and 5 show that the centro-affine surface area is really $GL$ invarinat, and that three point decimal accuracy will be sufficient to detect non $GL(3)$-equivalence between $4 \times 4 \times 3$ absolutely nonsingular tensors, whose elements consists of only -1,0,1.  The M1-P2-G7 method seems clearly the best for  discriminating the tensors relating to $GL$ nonequivalence. 
\section{Conclusion}
We treated the  $SL(4) \times SL(4) \times SL(3),$ $GL(4) \times GL(4) \times SL(3)$ or  $GL(4) \times GL(4) \times GL(3)$ non-equivalence problem of $4 \times 4 \times 3$ absolutely nonsingular tensors. We proposed a method to addres to the problem through the determinant polynomials. Furthermore we proposed to solve the problem by differential geometric $SL(3)$ or $GL(3)$ invariant of the constant surface of the determinant polynomials. From the numerical analysis by Mathematica,  it was shown that the stable values of invariants are obtainable numerically and also it was shown that the affine surface area and the centro-affine surface area are useful to detect the non-equivalence. This means that the algebraic problem: whether a system of algebraic equations with many variables can have real solutions or not, can be resolved by differential geometric methods.  It is a nice link between algebra and differential geometry.  Second, we investigated the spherical design method for calculating invariants.  At present, we think that the values given by  the adaptive lattice methods are more reliable than those given by the spherical design method.  In some future work, we expect to extend the result to more higher dimensional tensors and to know why the spherical design method does not give stable values of invariants.


\begin{thebibliography}{99}
\bibitem{Atkinson1}
M.D. Atkinson and S. Lloyd,
\newblock {\em Bounds on the ranks of some $3$-tensors},
\newblock Linear Algebra and its applications {\bf 31} (1980), 19--31.
\bibitem{Atkinson2}
M. D. Atkinson and M. Stephens,
        {\rm On the maximal multiplicative complexity 
             of a family of bilinear forms},
        Linear Algebra and its applications {\bf 27} (1979),1--8.
\bibitem{BB}E. Bannai and E. Bannai, {\em A survey of spherical designs and algebraic combinatorics on spheres.} Europian J. of Combinatorics,{\bf 30} (2009),1392--1425.
\bibitem{BL} W. Blaschke, Vorlesungen \"uber Differential geometrie II,
Springer Verlag, Berlin 1923.
\bibitem{Chen} L. Chen, Yi. Chen and Y. Mei, Classification of multipartite entanglement containing infinitely many kinds of states, 
Phys. Revs. A {\bf 74} (2006), no. 5, 052331, 1--12.
\bibitem{Common} P. Comon, J.M.F. ten Berge,  L.D. Lathauwer and J. Castaing,  Generic
and typical ranks of multi-way arrays, Linear Algebra Applications {\bf 430} (2009),  no. 11-12, 2997--3007. 
\bibitem{DGS}P. Delsarte, J.M. Goethals and J.J Seidel, {\em Spherical codes and designs}, Geom. Dedicata {\bf 6} (1977), 363--388.
\bibitem{Slone}  R. H. Hardin and Sloane, N.J.A., Spherical Designs \\ http://www2.research.att.com/~njas/sphdesigns/dim3/.
\bibitem{Hadwiger} H. Hadwiger, Vorlesungen $\ddot{u}$ber Inhalt, Oberfl$\ddot{a}$she und Isoperimetrie, Springer, Berlin,1957. 
\bibitem{Hug} D. Hug, Contributions to affine surface area, manuscripta mathematics {\bf 91} (1996), 283--301.   
\bibitem{Kolda} T.G. Kolda and B.W, Bader, Tensor decompositions and applications,
SIAM Review {\bf 51} (2009), no. 3, pp. 455--500.
\bibitem{Kruskal}J.B. Kruskal, Three-way arrays: rank and uniqueness of trilinear
decompositions, with application to arithmetic complexity and statistics,
Linear Algebra and Appl. {\bf 18} (1977), no. 2, 95--138.
\bibitem{Leicht} K. Leichtwess, {Affine Geometry of Convex bodies}, Johann Ambrosius.  Barth Verlag, Heidelberg (1998).
\bibitem{Ludwig0} M. Ludwig, A characterization of affine length and asymptotic approximation of convex discs, Abh. Math. semin. Univ. Hamb. {\bf 69} (1999), 75--78.
\bibitem{Ludwig2} M. Ludwig, {Valuations in the affine geometry of convex bodies}, Proceedings of the conference "Integral geometry and convexity", 
Wuhan 2004, World Scientific, Singapore (2006), 49--65.      
\bibitem{Ludwig1} M. Ludwig and M. Reitzner, A characterization of affine surface area, Adv. Math. {\bf 147} (1999), 138--172.
\bibitem{LudwigReit} M. Ludwig and M. Reitzner, A Classification of SL(n) invariant Valuations, Annals of Mathematics (2010), in press.\\
preprint, http://sites.google.com/site/monikaludwig/.
\bibitem{Maehara} K. Maehara, A list of absolutely nonsingular tensors with $-1,0,1$ elements for $4 \times 4 \times 3$ case. Preprint(2010).   
\bibitem{NS}K. Nomizu and T. Sasaki, Affine Differential Geometry, Cambridge Univ. Press, Cambridge (1994).
\bibitem{Lut} E. Lutwak,The Brunn-Minkovski-Firey theory II:Affine and geominimal surface areas, Adv.Math. {\bf 118} (1996), 244--294. 
\bibitem{Sakata1} T. Sakata, T. Sumi and M. Miyazaki, {Exceptional tensors with three slices and the positivity of its determinant polynomial}, Abstarct book of ISI, (2009),CPM37, Theoretical Statistics, 349.
\bibitem{Sakata2}
T. Sakata, T. Sumi, M. Miyazaki and K. Maehara, { Exceptional tensors of 3 x 4 x4 tensors and Hilbert 17 th problem},Abstract book of  Statistics, Probability, Operation Research, Computer Science and allied Areas, (2010), Complex Data Analysis and Modelling, 75--76.  
\bibitem {Sumi1} 
T. Sumi, M. Miyazaki, M. and T. Sakata,
\newblock {\em About the maximal rank of 3-tensors over the real
and the complex number field}, Ann. Inst. Stat. Math. {\bf 62} (2010), 807--822.
\bibitem{tenBerge} J.M.F. ten- Berge,  The typical rank of tall three-way arrays, Psychometrika {\bf 65} (2000),
no. 4, 525--532. 
\end{thebibliography}
\end{document}